\documentclass[11pt]{article}

\usepackage{amsmath}
\usepackage{amssymb}
\usepackage{fullpage}
\usepackage{color}
\usepackage{graphicx}
\usepackage{epsfig}
\usepackage{amsthm}
\usepackage{latexsym}
\usepackage{amssymb}
\usepackage{amsmath}
\usepackage{verbatim}
\usepackage{graphicx}
\usepackage{epsfig}
\usepackage{subfigure}
\usepackage{cite}
\usepackage{url}
\usepackage{tabularx}
\usepackage{algorithmic}
\usepackage{algorithm}
\usepackage{enumitem}

\usepackage[usenames,svgnames,xcdraw,table]{xcolor}
\definecolor{DarkBlue}{rgb}{0.1,0.1,0.5}
\definecolor{DarkGreen}{rgb}{0.1,0.5,0.1}
\usepackage[backref=page]{hyperref}
\hypersetup{
     colorlinks   = true,
     linkcolor    = DarkBlue, 
     urlcolor     = DarkBlue, 
	 citecolor    = DarkGreen 
}
\renewcommand*{\backref}[1]{}
\renewcommand*{\backrefalt}[4]{%
    \ifcase #1 (Not cited.)%
    \or        (Cited on page~#2)%
    \else      (Cited on pages~#2)%
    \fi}
\usepackage[capitalise,noabbrev]{cleveref}

\usepackage{thmtools,thm-restate}

\newtheorem{theorem}{Theorem}

\newtheorem{proposition}{Proposition}

\newtheorem{lemma}{Lemma}

\newtheorem{example}{Example}

\newcommand{\EF}{\mathrm{EF}}
\newcommand{\EFone}{\mathrm{EF1}}
\newcommand{\bEFone}{\mathrm{\mathbf{EF1}}}
\newcommand{\EFX}{\mathrm{EFX}}
\newcommand{\bEFX}{\mathrm{\mathbf{EFX}}}
\newcommand{\EFL}{\mathrm{EFL}}
\newcommand{\bEFL}{\mathrm{\mathbf{EFL}}}

\newcommand{\NSW}{\mathrm{NSW}}

\newcommand{\ALG}{\textsc{Alg}}
\newcommand{\MMS}{\mathrm{MMS}}

\newcommand{\CMMS}{\mathrm{CMMS}}
\newcommand{\bMMS}{\mathrm{\mathbf{MMS}}}

\DeclareMathOperator*{\argmax}{arg\,max}
\DeclareMathOperator*{\argmin}{arg\,min}

\begin{document}
\title{{\bfseries Fair Division Under Cardinality Constraints}}
\author{Arpita Biswas\thanks{Indian Institute of Science. {\tt arpitab@iisc.ac.in}}
\and Siddharth Barman\thanks{Indian Institute of Science. {\tt barman@iisc.ac.in}}
}
\date{}
\maketitle

\begin{abstract}
We consider the problem of \emph{fairly} allocating {indivisible} goods, among agents, under {cardinality constraints} and {additive} valuations. In this setting, we are given a partition of the entire set of goods---i.e., the goods are categorized---and a limit is specified on the number of goods that can be allocated from each category to any agent. The objective here is to find a \emph{fair} allocation in which the subset of goods assigned to any agent satisfies the given cardinality constraints. This problem naturally captures a number of resource-allocation applications, and is a generalization of the well-studied (unconstrained) fair division problem.  

The two central notions of fairness, in the context of fair division of indivisible goods, are \emph{envy freeness up to one good} ($\EFone$) and the (approximate) \emph{maximin share guarantee} ($\MMS$). We show that the existence and algorithmic guarantees established for these solution concepts in the unconstrained setting can essentially be achieved under cardinality constraints. Specifically, we develop efficient algorithms which compute $\EFone$ and approximately $\MMS$ allocations in the constrained setting. 


Furthermore, focusing on the case wherein all the agents have the same additive valuation, we establish that $\EFone$ allocations exist and can be computed efficiently even under \emph{laminar matroid constraints}. 
\end{abstract}

\section{Introduction}
A large body of recent work in algorithmic game theory, artificial intelligence, and computational social choice has been directed towards understanding the problem of allocating indivisible goods among agents in ``fair'' manner; see, e.g.,~\cite{brandt2016handbook} and~\cite{endriss2017trends} for excellent expositions. This recent focus on indivisible goods is motivated, in part, by applications (such as division of inheritance and partitioning computational resources in a cloud computing environment) which inherently entail allocation of goods/resources that cannot be fractionally allocated. In fact, algorithms developed for finding fair allocations of indivisible goods have been implemented in specific settings; for instance, Course Match~\cite{budish2016course} is employed for course allocation at the Wharton School in the University of Pennsylvania and the website Spliddit ({\tt www.spliddit.org})~\cite{spliddit} provides free, online access to fair division algorithms. 

Note that, though the theory of fair division is extensive, classical notions of fairness---such as \emph{envy freeness}\footnote{An allocation is said to be envy-free if each agent values her bundle at least as much as she values any other agent's bundle\cite{foley1967resource,varian1974equity}.}---typically address allocation of divisible goods (i.e., goods which can be fractionally allocated) and are not representative in the indivisible setting. For instance, while an envy-free allocation of divisible goods is guaranteed to exist~\cite{stromquist1980cut}, such an existence result does not hold for indivisible goods.\footnote{If we have a single indivisible good and two agents, then in any allocation, the losing agent is bound to be envious.}

Motivated by these considerations, recent results have formulated and studied solution concepts to address fair division of indivisible goods~\cite{budish2011combinatorial,procaccia2014fair,bouveret2014characterizing}. Arguably, the two most prominent notions of fairness in this context are (i) \emph{envy freeness up to one good} ($\EFone$) and (ii) the \textit{maximin share guarantee} ($\MMS$). These solution concepts were defined by Budish~\cite{budish2011combinatorial} and they, respectively, provide a cogent analogue of envy-freeness and \emph{proportionality}\footnote{An allocation is said to be proportionally fair among $n$ agents, if every agent gets a bundle of value at least $1/n$ times her value for the entire set of goods.} in the context of indivisible goods: 

\begin{enumerate}[label=(\roman*)]
\item An allocation is said to be $\EFone$ if every agent values her bundle at least as much as any other agent's bundle, up to the removal of the most valuable good from the other agent's bundle. $\EFone$ allocations are guaranteed to exist; by contrast, for indivisible goods envy-free allocations might not exist. Another attractive feature of $\EFone$ is that it is computationally tractable: even under combinatorial valuations, $\EFone$ allocations can be found efficiently~\cite{lipton-envy-graph}. Furthermore, under additive valuations, this notion of fairness is compatible with \emph{Pareto efficiency}~\cite{caragiannis2016unreasonable}.

\item An allocation is said to satisfy $\MMS$ if each agent receives a bundle of value at least as much as her \emph{maximin share}. These shares provide an agent-specific fairness threshold, and are defined as the maximum value that an agent can guarantee for herself if she were to partition the set of goods into $n$ bundles and then, from those bundles, receive the minimum valued one; here, $n$ is the total number of agents. That is, maximin share can be interpreted via an application of the standard \emph{cut-and-choose} protocol over indivisible goods: if agent $i$ is (hypothetically) asked to partition the set of goods into $n$ bundles and the remaining $(n-1)$ agents were to select their bundles before $i$, then a risk-averse agent would find a partition which maximizes the least valued bundle. Overall, this value that the agent $i$ can guarantee for herself is called the maximin share of $i$. Even though $\MMS$ allocations are not guaranteed to exist~\cite{procaccia2014fair,kurokawa2016can}, this notion admits efficient approximation guarantees. Specifically, under additive valuations, there exist polynomial-time algorithms for finding allocations wherein each agent receives a bundle of value at least $2/3$ times her maximin share~\cite{procaccia2014fair,amanatidis2015approximation,barman2017approximation}.\footnote{For additive valuations, Ghodsi et al.~\cite{ghodsi2017fair} provide an improved approximation guarantee of $3/4$.} Note that these results establish an absolute guarantee, i.e., they show that an approximately maximin fair allocation always exists. 

\end{enumerate}

The existence and computational results developed for $\EFone$ and $\MMS$ provide a sound understanding of fair division of indivisible goods. However, it is relevant to note that the vast majority of work in this thread of research is solely focussed on the unconstrained version of the problem.\footnote{The work of Bouveret et al.~\cite{fair-graph} along with \cite{gourves2017approximate} and \cite{gourves2014near} are notable exceptions. These results are discussed in Section \ref{sec:related}.} To address this limitation, and motivated by the fact that in many real-world settings the allocations are required to satisfy certain criteria, we study a  relevant, constrained version of the fair division problem.  

In particular, we consider a setting wherein the indivisible goods are categorized (i.e., we are given a partition of the set of goods) and a limit is specified on the number of goods that can be allocated from each category to any agent. Here, the objective is to find a fair allocation in which the subset of goods assigned to any agent satisfies the given cardinality constraints. We shall see that this corresponds to a fair allocation problem under a \emph{partition matroid constraint}.

The following stylized example---adapted from \cite{gourves2014near}---demonstrates the applicability of such constraints: A museum decides to open new branches, and thereby needs to transfer some of the exhibits from the main museum to the newly opened ones. The exhibits are categorized into, say, statues, paintings, and pottery. In addition, there is an upper limit on the number of exhibits that every newly opened branch can accommodate from each category. The question now is to find a feasible division of the exhibits which is fair to the curators of each of the new branches.\\

\noindent 
{\bf Our Contributions:} We show that the existence and algorithmic guarantees established for $\EFone$ and $\MMS$ allocations in the unconstrained setting can, in fact, be achieved under cardinality constraints as well. Specifically, we establish the following results under additive valuations and cardinality (partition matroid) constraints:

\begin{enumerate}
\item $\EFone$ allocations are guaranteed to exist. In particular, we develop a combinatorial algorithm which, for a given fair division instance with additive valuations and cardinality constraints, finds an $\EFone$ allocation in polynomial time (Theorem~\ref{theorem:EF1}).

\item In this constrained setting, a constant-factor approximate $\MMS$ allocation always exists and can be computed in polynomial time (Theorem~\ref{theorem:MMS}). Note that, in this setting, the value of the maximin share of each of the $n$ agents is obtained by considering only feasible allocations (see Equation (\ref{def:CMMS}) in Section~\ref{sec:prelims}). That is, here, the maximin share of an agent is defined to be maximum value that she can guarantee for herself if she were to partition the set of goods into $n$ bundles, each of which must satisfy the cardinality constraints, and from them receive the minimum valued one. 

\item We also consider fair division subject to laminar  matroid constraints (Section~\ref{sec:identical}) and show that if the agents have identical, additive valuations, then again an $\EFone$ allocation is guaranteed to exist (Theorem~\ref{theorem:identical}) and can be computed efficiently (Theorem~\ref{theorem:identical_algo}). 

\item Specific extensions of $\EFone$---in particular, \textit{envy-free up to the least valued good} ($\EFX$)\footnote{An allocation is said to be {envy-free up to the least valued good} ($\EFX$) if no agent envies any other agent after removing any positively valued good from the other agent's bundle \cite{caragiannis2016unreasonable}.} and \textit{envy-free up to one least-preferred good} ($\EFL$)\footnote{See Section~\ref{sec:EFL} for a definition.}---are not guaranteed to exist under cardinality constraints. In Section~\ref{sec:EFL} we provide a (constrained) fair division instance which does not admit an $\EFX$ or an $\EFL$ allocation.


\end{enumerate}

\subsection{Related Work} \label{sec:related}
The fairness notions of $\EFone$ and $\MMS$ were defined by Budish~\cite{budish2011combinatorial} (see also~\cite{moulin1990uniform}), and these solution concepts have been extensively studied since then. The existence of $\EFone$ allocations can be established via the \emph{cycle-elimination} algorithm of Lipton et al.~\cite{lipton-envy-graph}. Caragiannis et al.~\cite{caragiannis2016unreasonable} have shown that these two solution concepts are incomparable, i.e., one does not imply the other. 


Fair division instances wherein the agents' valuations are binary and additive always admit $\MMS$ allocations~\cite{bouveret2014characterizing}. However, Procaccia and Wang~\cite{procaccia2014fair} along with Kurokawa et al.~\cite{kurokawa2016can} provided intricate counterexamples to refute the universal existence of $\MMS$ allocations, even under additive valuations. This lead to a study of approximate maximin share allocations, $\alpha$-$\MMS$, where each agent receives a bundle whose value to her is at least $\alpha \in (0,1)$ times her maximin share. Procaccia and Wang~\cite{procaccia2014fair} developed an efficient algorithm to obtain a $2/3$-$\MMS$ allocation when the number of agents is constant. Later, Amanatidis et al.~\cite{amanatidis2015approximation} established that a $2/3$-$\MMS$ allocation can be computed in polynomial (even in the number of players) time.

More recently, constant-factor approximation guarantees for $\MMS$ in  settings wherein the valuations are not necessarily additive have been established in~\cite{barman2017approximation} and~\cite{ghodsi2017fair}. Specifically, for \emph{submodular} valuations Ghodsi et al.~\cite{ghodsi2017fair} have developed an efficient algorithm for computing $1/3$-$\MMS$ allocations.  


As mentioned previously, the work on fair division of indivisible goods is primarily confined to the unconstrained setting. Exceptions include the work of Bouveret et al.~\cite{fair-graph} and Ferraioli et al.~\cite{ferraioli2014regular}.  Bouveret et al.~\cite{fair-graph} consider fair division of goods which correspond to vertices of a given graph, and the problem is to fairly allocate a connected subgraph to each agent. It is shown in \cite{fair-graph} that $\MMS$ allocations might not exist for general graphs; however, it can be efficiently computed when the underlying graph is a tree. 

Ferraioli et al.~\cite{ferraioli2014regular} consider fair division problems where each agent must receive exactly $k$ goods, for a given integer $k$, and  provide an algorithm to efficiently compute a $1/k$-approximate $\MMS$ allocation. A different kind of matroid constraint is considered in \cite{gourves2017approximate} and \cite{gourves2014near}. In particular, the problem formulation in~\cite{gourves2017approximate} and~ \cite{gourves2014near} requires that the \emph{union} of all the allocated goods is an independent set of a given matroid. This setup is incomparable to the one considered in this paper. Specifically, while we ensure that in the computed fair allocation \emph{each} agent's bundle satisfies the partition matroid constraint and all the goods are allocated, these requirements are not imposed in \cite{gourves2017approximate} and \cite{gourves2014near}.

\section{Notation and Preliminaries}\label{sec:prelims}
An instance of the fair division problem comprises of a tuple $\left\langle[m],[n],(v_i)_{i\in[n]}\right\rangle$, where $[m]=\{1,2,\ldots,m\}$ denotes the set of indivisible goods, $[n]=\{1,2,\ldots, n\}$ denotes the set of agents, and $v_i$s specify the valuation (preferences) of the agents, $i \in [n]$, over the set of goods. Throughout, we will assume that, for each agent $i$, the valuation $v_i:2^{[m]} \mapsto \mathbb{R}_+$ is additive, i.e., for all agents $i \in [n]$ and subsets $S \subseteq [m]$, $v_i(S) = \sum_{g \in S}\ v_i \left (\{g\} \right)$. Also, we have $v_i\left(\{g\}\right)\geq 0$ for all $i \in [n]$ and $g\in [m]$. For ease of notation, we will use $v_i(g)$, instead of $v_i(\{g\})$, to denote the valuation of agent $i$ for a good $g$.

Write $\Pi_t(S)$ to denote the set of all $t$-partitions of a subset of goods $S \subseteq [m]$. An \textit{allocation}, $\mathcal{A}$, refers to an $n$-partition $\mathcal{A}=(A_1,A_2,\dots,A_n) \in \Pi_n([m])$, where $A_i$ is the subset of goods (bundle) allocated to agent $i$. In this work we focus on finding fair allocations which satisfy given cardinality constraints. Specifically, we are given a partition of the set of goods consisting of $\ell$ different ``categories'' $\{C_1, C_2, \ldots, C_\ell\}$, and associated with each category $h \in \{1,2,\ldots, \ell\}$ we have a (cardinality) threshold $k_h$. In this setup, an allocation $\mathcal{A}=(A_1,A_2,\dots,A_n)$ is said to be \emph{feasible} iff, for every bundle $A_i$ and category $h$, the cardinality constraint holds: $|A_i \cap C_h| \leq k_h$. Throughout, we will use $\mathcal{F}$ to denote the set of feasible allocations, $\mathcal{F}:=\{\mathcal{A}\in\Pi_{n}([m]) \ \mid \ |A_i \cap C_h| \leq k_h \text{ for all }  i\in[n] \text{ and } h\in[\ell]\}$. To ensure that $\mathcal{F}$ is nonempty, we require that the threshold $k_h\geq \frac{|C_h|}{n}$ for all categories $h\in[\ell]$. 

We will use $\left\langle[m],[n],(v_i)_{i\in[n]}, \mathcal{F} \right\rangle$ to denote an instance of the fair division problem subject to cardinality constraints. Overall, our goal is to find fair allocations contained in $\mathcal{F}$.\footnote{Note that the set $\mathcal{F}$ might be exponential in size, but it is specified in an efficient manner via the partition $\{C_1, C_2, \ldots, C_\ell\}$ and  thresholds $k_h$. Setting $\ell$=$1$, $C_1$=$ [m]$, $k_1$=$m$, we get $\mathcal{F}$=$\Pi_n([m])$. Hence, this formulation is a strict generalization of the unconstrained fair division problem.} In this work, we provide existential and algorithmic results for the following fairness notions:
\begin{itemize}
\item {Envy-free up to one good} ($\EFone$): In a fair division instance $\left\langle[m],[n],(v_i)_{i\in[n]}, \mathcal{F} \right\rangle$, an allocation $\mathcal{A}=(A_1,A_2,\dots,A_n)\in\mathcal{F}$ is said to be $\EFone$ iff for every pair of agents $i , j \in [n]$ there exists a good $g \in A_j$ such that $v_i(A_i) \geq v_i(A_j\setminus\{g\})$.
\item {Maximin Share Guarantee} ($\MMS$): Given an instance $\left\langle [m],[n],(v_i)_{i\in[n]}, \mathcal{F} \right\rangle$, the (constrained) maximin share of agent $i$ is defined as 
\begin{align}
\CMMS_i &:= \max_{(P_1,\ldots, P_n)\in \mathcal{F}}\  \min_{j\in [n]}\ v_i(P_j).\label{def:CMMS}
\end{align} 
An allocation $\mathcal{A}=(A_1,\ldots, A_n)\in\mathcal{F}$ is said satisfy $\MMS$ iff for all agents $i \in [n]$, we have $v_i(A_i) \geq \CMMS_i$. Since $\MMS$ allocations are not guaranteed to exist, the objective is to find feasible allocations $(A_1, \ldots, A_n) \in \mathcal{F}$ wherein each agent gets a bundle of value at least $\alpha \in (0,1]$ times $\CMMS_i$; with factor $\alpha$ being as large as possible. We call such allocations $\alpha$-$\MMS$. 
\end{itemize}

\section{Main Results}
The key results established in this paper are:

\begin{restatable}{theorem}{EFONE} 
\label{theorem:EF1}
\label{THEOREM:EF1}
Given any fair division instance $\left\langle[m],[n],(v_i)_{i\in[n]}, \mathcal{F} \right\rangle$ with additive valuations and cardinality constraints ($\mathcal{F} \neq \emptyset$), there exists a polynomial time algorithm for finding a feasible $\EFone$ allocation.
\end{restatable}

This theorem is established in Section~\ref{sec:EF1} and it implies that as long as the set of feasible allocations $\mathcal{F}$ is nonempty it admits an $\EFone$ (i.e., a fair) allocation. 


\begin{restatable}{theorem}{APXMMS}
\label{theorem:MMS}
\label{THEOREM:MMS}
Given any fair division instance $\left\langle[m],[n],(v_i)_{i\in[n]}, \mathcal{F} \right\rangle$ with cardinality constraints ($\mathcal{F} \neq \emptyset$) and additive valuations, a $1/3$-$\MMS$ allocation can be computed in polynomial time.
\end{restatable}

Analogous to the $\EFone$ case, this theorem provides an absolute, existence guarantee for approximate maximin fair allocations under cardinality constraints. A proof of this result appears in Section~\ref{sec:MMS}. 
 
We also establish that when the valuations are identical, then $\EFone$ allocations exist and are efficiently computable even under matroid constraints; see Theorem~\ref{theorem:identical} and Theorem~\ref{theorem:identical_algo} in Section~\ref{sec:identical}.

\section{$\bEFone$ Allocations Under Cardinality Constraints: Proof of Theorem~\ref{theorem:EF1}}\label{sec:EF1}
In the unconstrained setting, there exist efficient algorithms for finding $\EFone$ allocations; see, e.g., the \emph{cycle-elimination} algorithm of Lipton et al.~\cite{lipton-envy-graph} and the round-robin method by Caragiannis et al.~\cite{caragiannis2016unreasonable}. However, the allocations found by these algorithms are not guaranteed to satisfy the given cardinality constraints. We bypass this issue by developing a polynomial-time algorithm, $\ALG$~\ref{alg:EF1_partition}, for finding an allocation $\mathcal{A}$ which is not only $\EFone$, but also feasible, i.e., $\mathcal{A} \in \mathcal{F}$.

{
\begin{algorithm}
	{
		{{\bf Input :} A fair division instance $\left\langle[m],[n],(v_i)_{i},\mathcal{F}\right\rangle$ with additive valuations and cardinality constraints.\\
		{\bf Output:} A feasible $\EFone$ allocation}
		\caption{\ALG~\ref{alg:EF1_partition}
		}
		\label{alg:EF1_partition}
		\begin{algorithmic}[1]
		    \STATE Initialize allocation $\mathcal{A}^0=(A^0_1, \ldots, A^0_n)$ with $A^0_i~\leftarrow~\emptyset$ for each agent $i \in [n]$.	
		    \STATE Fix an (arbitrary) ordering of the agents $\sigma = \left(\sigma(1), \sigma(2),\ldots, \sigma(n)\right)$ 
			\FOR{$h =1 \mbox{ to } \ell$}	
				\STATE $\mathcal{B}^h\leftarrow$Greedy-Round-Robin$(C_h,[n],(v_i)_i, \sigma)$.
				\STATE Set $A^h_i\leftarrow A^{h-1}_i\cup B^h_i$ for all $i\in[n]$.
				\STATE Using Lemma \ref{lemma:envylipton}, update $\mathcal{A}^h=(A^h_1, \ldots, A^h_n)$ to obtain an acyclic envy graph $\mathcal{G}(\mathcal{A}^h)$. 
				\STATE Update $\sigma$ to be a topological ordering of $\mathcal{G}(\mathcal{A}^h)$.
			\ENDFOR 			
			\STATE Return $\mathcal{A}^{\ell}$.
		\end{algorithmic}}
	\end{algorithm}
}
\ALG~\ref{alg:EF1_partition} is based on an interesting modification of the round-robin algorithm: initially, $\ALG$~\ref{alg:EF1_partition} selects an arbitrary order (permutation) over the agents $ \sigma:=\left( \sigma(1),\ldots, \sigma(n) \right)$. It then picks an unallocated category $h$ and executes the \textit{Greedy-Round-Robin} algorithm ($\ALG$~\ref{alg:greedy}) with the $n$ agents, $|C_h|$ goods (from category $C_h$), and the selected order $\sigma$. 

$\ALG$~\ref{alg:greedy} follows the ordering $\sigma$ in a round-robin fashion (i.e., it selects agents, one after the other, from $\sigma(1)$ to $\sigma(n)$), and iteratively assigns to the selected agent an unallocated good from $C_h$ that she desires the most. Finally, it returns an allocation $\mathcal{B}^h \in \Pi_n(C_h)$. 

After allocating all the goods of a category $h$, Step $6$ of $\ALG$~\ref{alg:EF1_partition} creates an \textit{envy graph}\footnote{An envy graph, for an allocation $\mathcal{A}$, is a directed graph that captures the envy between agents in $\mathcal{A}$. Specifically, the nodes in the envy graph represent the agents and it contains a directed edge from $i$ to $j$ iff $i$ envies $j$, i.e., iff $v_i(A_i) < v_i(A_j)$.} $\mathcal{G}(\mathcal{A}^h)$. It was established in \cite{lipton-envy-graph} that one can always efficiently update a given partial allocation such that the resulting envy graph is acyclic:

\begin{lemma}\cite{lipton-envy-graph}
	\label{lemma:envylipton}
	Given a partial allocation $(A_1, \ldots, A_n) \in \Pi_n(S)$ of a subset of goods $S$, we can find another partial allocation $\mathcal{B}$=$(B_1, \ldots, B_n) \in \Pi_n(S)$ of $S$ in polynomial time such that \\
(i) The valuations of the agents for their bundles do not decrease: $v_i(B_i) \geq v_i(A_i)$ for all $i \in [n]$.\\
(ii) The envy graph $G(\mathcal{B})$ is acyclic.
\end{lemma}

{
	\begin{algorithm}
	{
		{{\bf Input :} An instance $\left\langle C,[n],(v_i)_{i}\right\rangle$ with additive valuations and an ordering $\sigma$ of $[n]$.\\
		 {\bf Output:} An allocation of the given $|C|$ goods among $n$ agents}
		\caption{Greedy-Round-Robin ($\ALG$~\ref{alg:greedy})}
		\label{alg:greedy}
		\begin{algorithmic}[1]
			\STATE Initialize bundle $B_i \leftarrow \emptyset$, for each agent $i \in [n]$, the set of unallocated goods $M \leftarrow C$, and $t\leftarrow 0$.			
			\WHILE{$M \neq \emptyset$}
				\STATE $t\leftarrow t+1$.
				\FOR{$i =1 \mbox{ to } n$}
			 		\STATE Set $g^t_{\sigma(i)} \in \argmax_{g \in M} v_{\sigma(i)}(g)$ 
					\STATE Update $B_{\sigma(i)} \leftarrow B_{\sigma(i)} \cup \{ g^t_{\sigma(i)} \}$ 
					\STATE Update $M \leftarrow M \setminus \{g^t_{\sigma(i)} \}$.
					\IF{$M == \emptyset$}	
						\STATE break;
					\ENDIF				
					\ENDFOR
			\ENDWHILE 			
			\STATE Return $\mathcal{B}=(B_1,\ldots,B_n)$.
		\end{algorithmic}}
	\end{algorithm}
}

Finally, $\sigma$ is updated to be a topological ordering of the acyclic directed graph $\mathcal{G}(\mathcal{A}^h)$. This new ordering is then used for the next category of goods. 

The feasibility of the computed allocation, $\mathcal{A}^\ell$, directly follows from the fact that (for each $h\in [\ell]$) $\ALG$~\ref{alg:greedy} distributes the $|C_h|$ goods evenly among $n$ agents. In particular, the round-robin nature of $\ALG$~\ref{alg:greedy} ensures that, $B_i^h$, the set of goods allocated to agent $i$ from category $h$ satisfies:\footnote{Recall that the underlying feasible set of allocations $\mathcal{F}$ is nonempty iff the integer limit $k_h \geq \frac{|C_h|}{n}$.}  
$|B_i^h|  \leq \left\lceil\frac{|C_h|}{n}\right\rceil\leq k_h$.

In Lemma~\ref{lemma:acyclic_Algo2} we show that, for each $h \in [\ell]$, the partial allocation obtained after the allocating the first $h$ categories, $\mathcal{A}^h$, is $\EFone$. Since Algorithm~\ref{alg:EF1_partition} allocates all the $m$ goods we get that the final allocation, $\mathcal{A}^\ell$, is in fact $\EFone$ as well. 

This, along with the observation that \ALG~\ref{alg:EF1_partition} runs in polynomial time completes the proof of Theorem~\ref{theorem:EF1}.

Next, we establish a proposition which will be used in the proof of Lemma~\ref{lemma:acyclic_Algo2}.

\begin{proposition}\label{lemma:acyclic_Algo1}
Given any fair division instance $\left\langle C,[n],(v_i)_{i}\right\rangle$ with additive valuations and an ordering of agents $\sigma$, the allocation $\mathcal{B}=(B_1, \ldots, B_n)$ obtained by $\ALG$~\ref{alg:greedy} satisfies the following properties:\\
$(1)$ For any two indices $i<j$, the agent $\sigma(i)$ does not envy agent $\sigma(j)$, i.e., $v_{\sigma(i)} (B_{\sigma(i)} ) \geq v_{\sigma(i)} (B_{\sigma(j)})$.\\
$(2)$ $\mathcal{B}$ is $\EFone$.
\end{proposition}
\begin{proof}
Write $T:=\lceil\frac{|C|}{n}\rceil$ to denote the total number of rounds of $\ALG$~\ref{alg:greedy}. In each round, for $i<j$, agent $\sigma(i)$ gets to choose her most desired good among the unallocated goods before agent $\sigma(j)$. Hence, if $g^t_{\sigma(i)}$ and $g^t_{\sigma(j)}$  denote the good assigned to agent $\sigma(i)$ and $\sigma(j)$, respectively, in the $t$th round, then $v_{\sigma(i)}(g^t_{\sigma(i)})~\geq~v_{\sigma(i)}(g^t_{\sigma(j)})$ for all $t \in \{1,\ldots,T\}$.
Since the valuations are additive, the stated property holds: $v_{\sigma(i)}(B_{\sigma(i)}) \geq v_{\sigma(i)}( B_{\sigma(j)})$.

It is known that the round-robin algorithm results in an $\EFone$ allocation~\cite{caragiannis2016unreasonable}, we repeat the argument for completeness: If index $i$ is less than $j$, then agent $\sigma(i)$ does not envy agent $\sigma(j)$. On the other hand, even if $i>j$ the good allocated to  agent $\sigma(i)$ in the $t$th round is of value (under $v_{\sigma(i)}$) no less than the good allocated to agent $\sigma(j)$ in the $(t+1)$th round: $
v_{\sigma(i)}(g^t_{\sigma(i)}) \geq v_{\sigma(i)}(g^{t+1}_{\sigma(j)})$ for all $t \in \{1,\ldots,T-1\}$.
Summing we get, $\sum_{t=1}^{T-1} v_{\sigma(i)}(g^t_{\sigma(i)}) \geq v_{\sigma(i)}(B_{\sigma(j)}) -v_{\sigma(i)}(g^1_{\sigma(j)})$. 
Thus, the allocation $\mathcal{B}$ is $\EFone$.
\end{proof}

\begin{lemma}\label{lemma:acyclic_Algo2}
$\ALG$~$\ref{alg:EF1_partition}$ returns an $\EFone$ allocation. 
\end{lemma}
\begin{proof}
We will show inductively that, for each $h \in [\ell]$, the partial allocation obtained after  allocating the first $h$ categories, $\mathcal{A}^h$, is $\EFone$. Hence, the returned allocation, $\mathcal{A}^\ell$ is $\EFone$ as well.  

Here, the base case ($h=1$) follows from Proposition~\ref{lemma:acyclic_Algo1}; since $\mathcal{A}^1 = \mathcal{B}^1$ and the proposition ensures that $\mathcal{B}^1$ is $\EFone$. 

By the induction hypothesis we have that $\mathcal{A}^h$ is $\EFone$ and, by construction, the corresponding envy graph $\mathcal{G}(\mathcal{A}^h)$ is acyclic. Next, we will show that this continues to hold for the next category $h+1$. 
Note that the ordering of agents for (executing \ALG~\ref{alg:greedy} over) the category $C_{h+1}$ is obtained by topologically sorting $\mathcal{G}(\mathcal{A}^h)$. Let $\pi$ be that topological ordering and write $\pi^{-1}(a)$ as the index of agent $a$ according to the ordering $\pi$. 

Now, if agent $a$ envies $b$ in $\mathcal{A}^h$, then there is a directed edge from $a$ to $b$ in $\mathcal{G}(\mathcal{A}^h)$ and, hence, $\pi^{-1}(a)<\pi^{-1}(b)$. In this case, Proposition~\ref{lemma:acyclic_Algo1} ensures that $a$ does not envy $b$ in $\mathcal{B}^{h+1}$, i.e., we have $v_a(B_a^{h+1}) \geq v_a(B_b^{h+1})$. The fact that $\mathcal{A}^h$ is $\EFone$ gives us $v_a(A_a^h) \geq v_a(A_b^h) - v_a(g)$, for some $g \in A_b^h$. Summing the last two inequalities and noting that $A^{h+1}_a = A^{h}_a \cup B_a^{h+1}$ and $A^{h+1}_b = A^{h}_b \cup B_b^{h+1}$ we get that $\mathcal{A}^{h+1}$ is $\EFone$ with respect to agent $a$ and $b$. 

The complementary case wherein $a$ does not envy $b$ in $\mathcal{A}^h$ is analogous, since $\mathcal{B}^{h+1}$ is guaranteed to be $\EFone$ in itself. Overall, this establishes the stated claim, that the final allocation $\mathcal{A}^\ell$ is $\EFone$.
\end{proof}

Lemma~\ref{lemma:acyclic_Algo2} along with the feasibility argument mentioned above (and the direct observation that \ALG~\ref{alg:EF1_partition} runs in polynomial time) proves Theorem~\ref{theorem:EF1}, which we restate below.

\EFONE*

\section{Non-Existence of $\bEFX$ and $\bEFL$ Allocations Under Cardinality Constraints}
\label{sec:EFL}

This section considers comparative notions of fairness which are stronger than $\EFone$. In particular, we show that allocations which satisfy \emph{envy-freeness up to the least valued good} ($\EFX$) are not guaranteed to exist under cardinality constraints.\footnote{It is relevant to note that the existence of $\EFX$ allocations in unconstrained, fair division instances (with additive valuations) remains an interesting, open question~\cite{caragiannis2016unreasonable}.} The non-existence result also holds for \textit{envy-free up to one less-preferred good} ($\EFL$) allocations; defined in~\cite{barman2018groupwise}, this solution concept is weaker than $\EFX$, but stronger than $\EF1$. Formally, 

\begin{itemize}
\item  Envy-free up to the least valued good ($\EFX$)~\cite{caragiannis2016unreasonable}: An allocation $\mathcal{A}=(A_1,A_2,\dots,A_n)$ is said to be $\EFX$ iff for every pair of agents $i ,j \in [n]$ and for all goods $g \in A_j \cap \{ g' \in [m] \mid v_i(g') >0 \}$ (i.e., for all goods $g$ in $A_j$ which are positively valued by agent $i$) we have $v_i(A_i)\geq v_i(A_j\setminus\{g\})$.
\item Envy-free up to one less-preferred good ($\EFL$)~\cite{barman2018groupwise}: An allocation $\mathcal{A}=(A_1,A_2,\dots,A_n)$ is said to be $\EFL$ iff for every pair of agents $i , j  \in [n]$ at least one of the following conditions hold:
\begin{itemize}
\item $A_j$ contains at most one good which is positively valued by $i$, i.e., $|A_j \cap \{ g' \mid v_i(g') > 0 \} | \leq 1 $  
\item There exists a good $g \in A_j$ such that $v_i(A_i) \geq v_i(A_j\setminus\{g\})$ and $v_i(A_i) \geq v_i(\{g\})$.
\end{itemize}
\end{itemize}

Note that, by definition, an $\EFX$ allocation is necessarily $\EFL$ and, similarly, an $\EFL$ allocation is always $\EFone$. 

The universal existence of $\EFL$ allocations in unconstrained fair division instances (with additive valuations) was established in~\cite{barman2018groupwise}. By contrast, the following example shows that, under cardinality constrains, $\EFL$ allocations do not always exist. This also implies the non-existence of $\EFX$ under cardinality constraints.

\begin{example}
\label{example:noEFL}
Consider an instance with $2$ agents and $4$ goods. Here, the bundle of each agent is constrained to contain at most $2$ goods. That is, we have a (uniform matroid) single category $C_1$, which contains all the goods, and the threshold $k_1 = 2$. 

The valuations of both the agents are identical and as follows: $v(g_1)=50\ $ and $\ v(g_2)=v(g_3)=v(g_4)=1$. In this example, no feasible allocation will ensure $\EFL$.
\end{example}

In this example any feasible allocation will allocate exactly two goods to each agent. Therefore, one of the agents, say $a_1$, will receive a bundle of value $51$ (containing the high-valued good $g_1$ and a low-valued good), while the other agent, $a_2$, will receive a bundle of value $2$ (containing two low-valued goods). Since agent $a_2$ envies $a_1$, $\EFL$ can be ensured only if there exists a good of value at most $2$ in $a_1$'s bundle which, when removed, eliminates $a_2$'s envy. However, such a good does not exist in $a_1$'s bundle and, hence, no feasible allocation is $\EFL$.

\section{$\bMMS$ Under Cardinality Constraints: Proof of Theorem~\ref{theorem:MMS}}\label{sec:MMS}
To obtain an $\alpha$-approximate maximin share allocation ($\alpha$-$\MMS$) under cardinality constraints we define a nonnegative, monotone, submodular function $F_i(\cdot)$, for each agent $i\in [n]$, which satisfies  
\begin{align} 
\underset{(P_1, \ldots, P_n) \in \mathcal{F}}{\max} \  \ \underset{j\in [n]}{\min}\ v_i(P_j)
 & = \underset{ (P_1, \ldots, P_n) \in\Pi_n([m])}{\max} \  \ \underset{j\in [n]}{\min}\ F_i(P_j).\nonumber
\end{align}
Hence, a fair division problem under additive valuations and cardinality constraints reduces to an unconstrained fair division problem under monotone, submodular valuations $(F_i)_i$. Recall that a function $F_i(\cdot)$ is said to be monotone and submodular iff for all subsets $A\subseteq B\subseteq[m]$ and $g\in[m]\setminus B$, we have $F(A)\leq F(B)$ and $F(A\cup\{g\})-F(A)\geq F(B\cup\{g\})-F(B)$.

We define the function $F_i(\cdot):2^{[m]}\rightarrow\mathbb{R}$ as follows $F_i(S) := \displaystyle \sum_{h\in[\ell]} f_i^h(S)$
where, 
\begin{align*}
f_i^h(S) & :=
  \begin{cases}
    \displaystyle \sum_{g\in S \cap C_h} v_i(g), & \text{if } |S\cap C_h|\leq k_h\\
    \displaystyle \sum_{g \in \mathtt{Top}_i^{k_h}(S \cap C_h)} v_i(g), & \text{otherwise}
  \end{cases}
  \end{align*}
  
Here, $\mathtt{Top}_i^k(T)$ denotes the set of the $k$ most valued (by agent $i$) goods contained in $T$. The following lemma asserts that $F_i$s are submodular. 
  
\begin{lemma}\label{lemma:mms_submodular}
For each agent $i\in[n]$ the function $F_i(\cdot{})$ (defined above) is monotone, nonnegative, and submodular.
\end{lemma}
\begin{proof}
Since, for each agent $i \in [n]$, the valuation $v_i(\cdot)$ is nonnegative, the function $F_i(\cdot)$ is nonnegative as well. Furthermore, for any agent $i \in [n]$ and category $h \in [\ell]$, the function $f_i^h$ is monotone; it either increases or remains constant after a good $g \in [m]$ is included in any bundle $S$. Hence, $F_i(\cdot)$ is monotone. 

Next, we will show that for any agent $i \in [n]$ the function $F_i$ is submodular, i.e., it satisfies the following inequality, for any two subsets $S\subseteq T \subset [m]$ and good $g \notin T$.
\begin{align}
F_i(S\cup\{g\})-F_i(S)\geq F_i(T\cup\{g\})-F_i(T).\label{eq:submodularity}
\end{align}


Write $h$ to denote the category of the good $g$, i.e., $g \in C_h \setminus T$. Since including $g$ in the subsets only effects the value of $f_i^h$, we have  $F_i(S\cup\{g\})-F_i(S) = f_i^h(S \cup\{g\} ) - f_i^h(S)$ and 
$F_i(T\cup\{g\})-F_i(T) = f_i^h(T \cup\{g\} ) - f_i^h(T)$. 

Note that if $|T\cap C_h|< k_h$, then we have $f_i^h(S \cup\{g\} ) - f_i^h(S) = v_i(g)$ and $f_i^h(T \cup\{g\} ) - f_i^h(T) = v_i(g)$. Therefore, if the cardinality of $T \cap C_h$ is strictly less that $k_h$, then the inequality (\ref{eq:submodularity}) is in fact tight. 
Otherwise, if $|T\cap C_h|\geq k_h$, then we have the following two cases
\begin{enumerate} 
\item[] Case 1: $g \notin \mathtt{Top}_i^{k_h}\left( (T \cup \{g \}) \cap C_h \right)$. Here, we have $f_i^h(T\cup\{g\})=f_i^h(T)$. Furthermore, the monotonicity of function $f_i^h(\cdot)$ ensures that $f_i^h(S \cup\{g\} ) -  f_i^h(S)\geq 0 $. Therefore, inequality~(\ref{eq:submodularity}) holds.  
\item[] Case 2: $g \in \mathtt{Top}_i^{k_h} \left( (T \cup \{g \}) \cap C_h \right)$. In this case, $f_i^h(T \cup\{g\} ) - f_i^h(T)= v_i(g)-v_i(g')$, where $g' \in \underset{\hat{g}\in \mathtt{Top}_i^{k_h}(T\cap C_h)}{\argmin}\ v_i(\hat{g})$.
\begin{itemize}
\item Now, if $|S\cap C_h|<k_h$, then $f_i^h(S \cup\{g\} ) - f_i^h(S)= v_i(g)\geq v_i(g)-v_i(g')$ and, hence, inequality (\ref{eq:submodularity}) holds. 
\item Otherwise, if $|S\cap C_h|\geq k_h$, we have $f_i^h(S \cup\{g\} ) - f_i^h(S)=v_i(g)-v_i(g'')$, where $g'' \in \underset{\hat{g}\in \mathtt{Top}_i^{k_h}(S\cap C_h)}{\argmin}\ v_i(\hat{g})$. Since $S\subseteq T$, we have $\underset{\hat{g}\in \mathtt{Top}_i^{k_h}(T\cap C_h)}{\min}\ v_i(\hat{g})\leq \underset{\hat{g}\in \mathtt{Top}_i^{k_h}(S\cap C_h)}{\min}\ v_i(\hat{g})$. That is, $v_i(g')\leq v_i(g'')$ and, again, inequality (\ref{eq:submodularity}) is satisfied. 
\end{itemize}
\end{enumerate}
\end{proof}

Recall that the set of feasible allocations $\mathcal{F}$ is nonempty if and only if the cardinality thresholds $k_h\geq \frac{|C_h|}{n}$ for all categories $h\in[\ell]$. Assuming nonempty set of feasible allocations, we now establish that the maximin value over all feasible allocations using valuations $v_i(\cdot)$ is equivalent to the maximin value over all possible $[n]$-partitions of $[m]$ using the function $F_i(\cdot)$.

\begin{lemma}\label{lemma:mms_equality}
If for all the categories $h \in [\ell]$, the cardinality threshold satisfies $k_h \geq \frac{|C_h|}{n}$ (i.e., if the set of feasible allocations $\mathcal{F}$ is nonempty), then the following equality holds for each agent $i \in [n]$:
\begin{align*}
\underset{(P_1, \ldots, P_n) \in \mathcal{F}}{\max}\ \ \ \underset{j\in [n]}{\min}\ v_i(P_j)
 & = \underset{ (P_1, \ldots, P_n) \in\Pi_n([m])}{\max}\ \ \ \underset{j\in [n]}{\min} \ F_i(P_j)
\end{align*}
\end{lemma}
\begin{proof}
We fix an agent $i \in [n]$ and first show that the left-hand side of the stated equality is upper bounded by the right-hand side. 
Write $(A^*_1,\ldots,A^*_n) \in \argmax_{(P_1, \ldots, P_n) \in \mathcal{F}} \ \  \min_{j \in [n]} v_i(P_j)$. Since $|A^*_j\cap C_h|\leq k_h $ for all $h $ and $j$, we have $F_i(A^*_j) =v_i(A^*_j)$ for all $j$. Now, the inequality $\min_{j\in[n]} F_i(A^*_j) \leq \max_{(P_1, \ldots, P_n) \in \Pi_n([m])} \ \min_{j \in [n]} \  F_i(P_j)$ establishes the upper bound. 

We complete the proof by showing that an inequality holds in the other direction as well. 
Write  $(B^*_1,\ldots,B^*_n) \in \argmax_{(P_1, \ldots, P_n) \in \Pi_n([m])} \ \min_{j \in [n]} \  F_i(P_j)$.  Say $(B^*_1,\ldots,B^*_n)\notin\mathcal{F}$, then there exists an index $b \in [n]$ and $h \in [\ell]$ such that $|B^*_b\cap C_h|>k_h$. Since $k_h\geq \frac{|C_h|}{n}$, an averaging argument implies that there exists another index $a\in[n]$ for which $|B^*_a\cap C_h|<k_h$. 
Now, consider the lowest valued (by agent $i$) good $g \in B^*_b \cap C_h$. Note that $F_i(B^*_b)= F_i(B^*_b\setminus\{g\})$ and $F_i(B^*_a\cup\{g\})\geq F_i(B^*_a)$. Hence, we can iteratively perform such swaps till all the cardinality constraints are satisfied. That is, we can obtain an allocation, say $\mathcal{B}'=(B'_1,\ldots, B'_n) \in \mathcal{F}$, which is feasible and satisfies $F_i (B'_j) \geq F_i (B^*_j) $ for all $j \in [n]$. The feasibility of $\mathcal{B}'$ ensures that the following equality holds for all $j$: $F_i(B'_j)=v_i(B'_j)$.  Therefore, $\min_{j\in[n]} F_i(B'_j) = \min_{j\in[n]} v_i(B'_j) \leq \max_{(P_1, \ldots, P_n) \in \mathcal{F}} \ \  \min_{j \in [n]} v_i(P_j)$. Hence, the left-hand side of the equality stated in the lemma is at least as much as the right-hand side. 
\end{proof}

Under submodular valuations a $1/3$-$\MMS$ allocation can be computed efficiently in the unconstrained setting~\cite{ghodsi2017fair}. Therefore, for an unconstrained fair division instance over the $m$ goods and with submodular valuations of the $n$ agents as $(F_i)_i$, we can, in polynomial time, find an allocation $\mathcal{A}=(A_1,\ldots,A_n) \in \Pi_n([m])$ which is $1/3$-$\MMS$. Employing a swap argument--- similar to the one used in the  proof of Lemma~\ref{lemma:mms_equality}---we can efficiently convert $\mathcal{A}$ into a feasible allocation $\mathcal{B} \in \mathcal{F}$ which satisfies $F_i(B_i) \geq F_i(A_i)$, for all $i\in[n]$. That is, for the unconstrained instance, $\mathcal{B}$ is a $1/3$-$\MMS$ allocation as well. In addition, the feasibility of $\mathcal{B}$ implies that $F_i(B_i) = v_i(B_i)$, for all $i$. Overall, via Lemma~\ref{lemma:mms_equality} (i.e., the fact that the maximin shares of each agent $i$ in the constructed unconstrained instance is equal to the underlying $\CMMS_i$ value), we get that $\mathcal{B}$ is a feasible, $1/3$-$\MMS$ allocation for the constrained instance. This completes the proof of Theorem~\ref{theorem:MMS}, which is restated below. 

\APXMMS*

\section{Identical Valuations and Matroid Constraint}\label{sec:identical}
This section shows that if the additive valuations of the agents are identical (i.e., $v_i = v_j$ for all $i,j \in [n]$), then an $\EFone$ allocation is guaranteed to exist even under a laminar matroid constraint. 

Matroids have been studied extensively in mathematics and computer science; see, e.g.~\cite{oxley1992matroid}. These structures provide an encompassing framework for representing combinatorial constraints; in particular, the cardinality constraints considered in the previous sections correspond to a particular matroid, called the partition matroid. Formally, a \emph{matroid} is defined as a pair $ ([m], \mathcal{I})$ where $[m]$ is the ground set of elements and $\mathcal{I}$---referred to as \emph{independent sets}---is a nonempty collection of subsets of $[m]$  that satisfies: (i) Hereditary property: If $B \in \mathcal{I}$ and $A \subset B$, then $A \in \mathcal{I}$, and (ii) Independent Set Exchange: If $A,B \in \mathcal{I}$ and $|A| < |B|$, then there exist an element $ x \in B \setminus A$ such that $A \cup \{x\} \in \mathcal{I}$.  

We consider a fair division instance $\langle[m],[n], (v_i)_{i\in[n]},\mathcal{M}\rangle$ where $\mathcal{M}$ denotes the set of all allocations which satisfy the underlying matroid constraint, i.e., allocations whose constituent bundles are independent, $\mathcal{M} :=\{\mathcal{A} = (A_1, \ldots, A_n) \in\Pi_n([m]) \mid \ A_i \in \mathcal{I}$ for all $i\in[n]\}$. One such construct is called \textit{laminar} matroids which generalizes the cardinality constrained setting.  A \textit{Laminar matroid} is a 2-tuple $([m], \mathcal{I})$, where the goods are categorized as $\ell$ sets $C_1,C_2,\ldots , C_{\ell}$ such that $\displaystyle \cup_{h=1}^{\ell} C_h = [m]$.   Additionally, for every pair $x,y\in\{1,\ldots,\ell\}$, only one of the following holds true: (i)~$C_x\subset C_y$ or, (ii)~$C_y\subset C_x$ or, (iii)~ $C_x\cap C_y=\emptyset$. Each category $h$ has an upper threshold $k_h$. A set $I \subseteq [m]$ is \textit{independent}
if and only if $|I \cap C_h| \leq k_h$ for each $h\in\{1,\ldots,\ell\}$.

 The main result of this section is as follows 

\begin{theorem}\label{theorem:identical}
Every fair division instance $\left\langle[m],[n],(v_i)_{i\in[n]}, \mathcal{M}\right\rangle$ under additive, identical valuations and laminar matroid constraint, $\mathcal{M} \neq \emptyset$ , admits an $\EFone$ allocation.
\end{theorem}
\begin{proof}
We establish the result by showing that an allocation which maximizes the \textit{Nash Social Welfare} over the set $\mathcal{M}$ is necessarily $\EFone$. For  an allocation $\mathcal{A}=(A_1, \ldots, A_n)$, the {Nash Social Welfare} ($\NSW$) is defined as the geometric mean of the agents' valuations, $\NSW(\mathcal{A}):=(\prod_{i\in[n]}\ v_i(A_i))^{1/n}$. 

Let $v(\cdot)$ denote identical, additive valuations of the agents. We assume $v(g)>0$ for all $g\in[m]$ and consider an optimal allocation $\mathcal{A}\in \underset{\mathcal{B}\in[m]}{\argmax}\ \NSW(\mathcal{B})$, which satisfies $\NSW(\mathcal{A})>0$.\footnote{If the optimal $\NSW$ over $[m]$ is zero, then it must be the case that we have less than $n$ goods. For such an instance, an allocation wherein each agent gets at most one good is both feasible and $\EFone$.}  We will prove that $\mathcal{A}$ is $\EFone$. Since $\mathcal{A} \in \mathcal{M}$, the stated claim follows. 

Say, for contradiction, that $\mathcal{A}$ is not an $\EFone$ allocation, then we will show that there exists another allocation $\mathcal{A}' \in \mathcal{M}$ along with agents $i$ and $j$, such that $A'_h=A_h$, for all other agents $h\in [n]\setminus\{i,j\}$, and $\min\{v(A'_i), v(A'_j)\}>\min\{v(A_i), v(A_j)\}$. The last inequality and the fact that $v(A'_i) + v(A'_j) = v(A_i) + v(A_j)$ imply $v(A'_i) \cdot v(A'_j) > v(A_i) \cdot v(A_j)$. Therefore, we get $\NSW(A')>\NSW(A)$, which contradicts the optimality of $\mathcal{A}$. 

By definition, if $\mathcal{A}$ is not $\EFone$, then there exists a pair of agents $i,j$ such that 
\begin{align*}
({\mathbf{E}}): & \ \ v(A_i)<v(A_j)-v(g) \qquad \text{ for all } g\in A_j.
\end{align*}

Note that if there exists a good $g$ in $A_j$ such that $A_i \cup \{g\}$ is independent (i.e., $A_i \cup \{g\} \in \mathcal{I}$), then swapping $g$ from $A_j$ to $A_i$ gives us the desired allocation $\mathcal{A}'$ (with a strictly higher $\NSW$ than $\mathcal{A}$). Hence, we analyze the case in which no such good exists. In particular, we have $A_i\cup \{g\}\notin \mathcal{I}$ for all $g\in A_j$. This condition implies that $|A_i| \geq |A_j|$; otherwise, the Independent Set Exchange property of matroids would ensure the existence of a good $g\in A_j$ such that $A_i\cup\{g\}\in\mathcal{I}$.

Write $t:=|A_j|$. 
The independence of $A_i$ and $A_j$  along with the inequality $|A_i| \geq |A_j|=t$ imply that there exist $t$ component-wise distinct pairs of goods $\{ (g^i_z, g^j_z) \in (A_i, A_j) \mid z \in [t]\}$ such that $(A_i \setminus\{g^i_z\})\cup\{g^j_z\}$ and $(A_j\setminus\{g^j_z\})\cup\{g^i_z\}$ are independent for all $z \in [t]$. This fact is established in \cite{goemans2009lecture} for the case wherein $|A_i|=|A_j|=t$.  Lemma \ref{lemma:matching} (proved below) complements this result and guarantees the existence of the relevant $t$ component-wise distinct pairs of goods, $\{ (g^i_z, g^j_z) \in (A_i, A_j) \mid z \in [t]\}$, even when $|A_i| > |A_j|$. Next, we will complete the proof of the theorem using the pairs $\{ (g^i_z, g^j_z) \in (A_i, A_j) \mid z \in [t]\}$.



Since the pairs in $\{ (g^i_z, g^j_z) \in (A_i, A_j) \mid z \in [t]\}$ are distinct, $\{g^j_z \mid z \in [t]\} = A_j$ and $\{g^i_z \mid z \in [t]\} = \hat{A}_i \subseteq A_i$. The envy between agent $i$ and $j$ (i.e., $v(A_j)>v(A_i)\geq v(\hat{A}_i)$) implies that there exists index $y \in [t]$ for which $v(g^j_y)>v(g^i_y)$.

Consider the allocation $\mathcal{A}'$ wherein  $A'_i=(A_i\setminus\{g^i_y\})\cup\{g^j_y\}$, $A'_j=(A_j\setminus\{g^j_y\})\cup\{g^i_y\}$, and $A'_h=A_h$ for all other agents. Note that all the bundles in $\mathcal{A}'$ are independent, $\mathcal{A}' \in \mathcal{M}$, and  $v(A'_i)=v(A_i)+v(g^j_y)-v(g^i_y)>v(A_i)$. Furthermore, $v(A'_j) = v(A_j)-v(g^j_y)+v(g^i_y)>v(A_i)$; the last inequality follows from $(\mathrm{\mathbf{E}})$. Hence, $\min\{v(A'_i), v(A'_j)\}>v(A_i)=\min\{v(A_i), v(A_j)\}$, which implies that the $\NSW$ of $\mathcal{A}'$ is higher than that of $\mathcal{A}$. This, by contradiction, completes the proof. 
\end{proof}



\begin{lemma}\label{lemma:matching}
Let $([m],\mathcal{I})$ be a laminar matroid with independent subsets $I , J \in \mathcal{I}$ which satisfy $I \cap J=\emptyset$ and $|I| \geq |J|$. Then, there exists a one-to-one map $\mu: J \mapsto I$ such that swapping any $j \in J$ with $\mu(j) \in I$ leads to independent subsets, i.e., for any element $j \in J$, both $(J \setminus \{ j \}) \cup \{ \mu(j) \}$ and $(I \setminus \{ \mu(j) \}) \cup \{ j \}$ belong to $\mathcal{I}$. 
\end{lemma}

\begin{proof}
Let $C_1,\ldots, C_{\ell}$ be the laminar categories (or sets)---for any two indices $x,y\in\{1,\ldots,\ell\}$, only one of the following holds true: (i)~$C_x\subset C_y$, (ii)~$C_y\subset C_x$ or, (iii)~ $C_x\cap C_y=\emptyset$. Thus, the categories can be ordered in a way such that, for any $x<y$, either $C_x\subset C_y$ or $C_x\cap C_y=\emptyset$. This can be achieved by bottom-up level-order tree traversal\footnote{For example, if there are four categories such that $C_1\subset C_2 \subset C_4$, $C_3 \subset C_4$ and $C_3\cap C_2=\emptyset$, then the corresponding graph has three levels: $C_4$ in level $1$, $C_2$ and $C_3$ in level $2$ and $C_1$ is in level $3$. The bottom-up level-order tree traversal would give $(C_1, C_2, C_3, C_4)$ or $(C_1, C_3, C_2, C_4)$.} of a directed graph $G = (V,E)$ where each vertex $v\in V$ corresponds to a category $C_v$, and each directed edge $(y,x)\in E$ corresponds to the fact that $C_x\subset C_y$ and that no other index $z$ satisfies $C_x\subset C_z \subset C_y$. For ease of representation, we assume that $C_1,\ldots, C_{\ell}$ follow such an ordering. We use $\Psi(C_y)$ to denote the set of immediate sub-categories of $C_y$, i.e., $C_x\in \Psi(C_y)$ implies there is an edge $(y,x)\in E$ in the directed graph $G(V,E)$. Note that, for any vertex $x$, $|\Psi^{-1}(C_x)|\leq 1$ (no vertex has more than one incoming vertex). 

Now, let $I , J \in \mathcal{I}$ be two independent sets such that $I \cap J=\emptyset$ and $|I| \geq |J|$. 
We establish the existence of the desired mapping $\mu:J\mapsto I$ in a constructive manner. At any step, when we assign $\mu(j)=i$, we call the element $j$, as well as $i$, as a \textit{matched} element. For any subset $S\subseteq J$, we write $\mu(S):=\cup_{j\in S}\ \mu(j)$ and thus, $|\mu(S)| = |S|$. During the construction, for each category $C_x$, we maintain a set $U_x$ of unmatched elements (of $I$ or $J$) belonging to the category $C_x$. The following invariants are maintained on $\mu$ and $U_x$ for each category $C_x$, based on the value $\delta_x:=|C_x\cap I|-|C_x\cap J|$ (we will prove that the invariants continue to hold during the construction). 
\begin{enumerate}
\item When $\delta_x=0$, all elements are matched, i.e., $U_x=\emptyset$ and the constructed map $\mu$ contains a mapping from each element in $C_x\cap J$ to an unique element in $C_x\cap I$. 
\item When $\delta_x>0$, the set of elements that remains unmatched belongs to $C_x\cap I$, i.e., $U_x\subset C_x\cap I$, $|U_x|=\delta_x$, and for the constructed map $\mu$, we have $U_x = (C_x\cap I)\setminus \mu(C_x\cap J)$.
\item When $\delta_x<0$, the set elements that remains unmatched belongs to $C_x\cap J$, i.e., $U_x\subset C_x\cap J$, $|U_x|=-\delta_x$, and for the constructed map $\mu$, we have $U_x = (C_x\cap J)\setminus \mu^{-1}(C_x\cap J)$.
\end{enumerate}

We obtain the required $\mu$ by repeating the following process for each category $C_x$, in increasing order of their index $x$ from $1$ to $\ell$.
\begin{enumerate}
\item If $\Psi(C_x)= \emptyset$, then the elements of $C_x$ have not been processed before. For such an $x$, we compute $\delta_x = |C_x\cap I| - |C_x\cap J|$ and execute the following:
\begin{itemize}
\item if $\delta_x=0$, then, for each $j\in (C_x\cap J)$, we set $\mu(j)$ to some unique $i\in(C_x\cap I)$, and set $U_x=\emptyset$.
\item if $\delta_x>0$, then, for each $j\in (C_x\cap J)$, we set $\mu(j)$ to some unique $i\in(C_x\cap I)$. However, there are $\delta_x$ elements from $C_x\cap I$ that remain unmatched and we store them in $U_x$. Thus, $|U_x|=\delta_x$.
\item if $\delta_x<0$, then, each $i\in(C_x\cap I)$ is assigned to some unique $j\in(C_x\cap J)$ such that $\mu(j)=i$. There are $\delta_x$ elements from $C_x\cap J$ that remain unmatched and we store them in $U_x$. Thus, $|U_x|=-\delta_x$. 
\end{itemize}
\item If $\Psi(C_x)\neq \emptyset$, then some elements of $C_x$ may have been already processed. Each processed element must belong to some $C_y\in \Psi(C_x)$ and must have been assigned to either $\mu$ or to $U_y$. The unprocessed (and hence unmatched) elements of $C_x$ is denoted as $\psi'(C_x):=C_x\setminus \left(\underset{y\in \Psi(C_x)}{\cup} C_y\right)$.  Assuming $\delta_x = |C_x\cap I| - |C_x\cap J|$, we group the categories in the set $\Psi(C_x)$ as:  
$\Delta^+=\{y: C_y\in \Psi(C_x)\ \mbox{ and }\ \delta_y>0\}$, 
$\Delta^-=\{y: C_y\in \Psi(C_x)\ \mbox{ and }\ \delta_y<0\}$, and
$\Delta^0=\{y: C_y\in \Psi(C_x)\ \mbox{ and }\ \delta_y=0\}$. 
Then, we obtain the following,
\begin{equation}
\delta_x = \sum_{y\in\Delta^+}\delta_y + \sum_{y\in\Delta^-}\delta_y + \sum_{y\in\Delta^0}\delta_y + |\psi'(C_x)\cap I| - |\psi'(C_x)\cap J|.\label{eq:delta}
\end{equation}

\begin{itemize}
\item if $\delta_x=0$, Equation~\ref{eq:delta} implies
\begin{eqnarray}
&& \sum_{y\in\Delta^+}\delta_y + |\psi'(C_x)\cap I| = \sum_{y\in\Delta^-}(-\delta_y) + |\psi'(C_x)\cap J|\nonumber \\
&\Rightarrow & \sum_{y\in\Delta^+}|U_y| + |\psi'(C_x)\cap I| = \sum_{y\in\Delta^-}|U_y| + |\psi'(C_x)\cap J|.\label{eq:transfer}
\end{eqnarray}
Equation~\ref{eq:transfer} shows that the number of unmatched elements of $I$ is equal to the number of unmatched elements of $J$. Also, for each $y\in\Delta^+$, the invariant on $U_y$ (which holds inductively) ensures that $U_y\subseteq C_y\cap I\subset C_x\cap I$ (since $C_y\subset C_x$) and, for each $y\in\Delta^-$, we have $U_y\subseteq C_y\cap J\subset C_x\cap J$. Therefore, for each $j\in \left(\cup_{y\in\Delta^-} U_y\right)\cup (\psi'(C_x)\cap J)$, we set $\mu(j)$ to some unique $i\in \left(\cup_{y\in\Delta^+} U_y\right)\cup (\psi'(C_x)\cap I)$. Since no elements of category $C_x$ remains unmatched after this point, we set $U_x = \emptyset$. 
\item if $\delta_x>0$, Equation~\ref{eq:delta} implies
\begin{eqnarray}
\sum_{y\in\Delta^+}|U_y| + |\psi'(C_x)\cap I| > \sum_{y\in\Delta^-}|U_y| + |\psi'(C_x)\cap J|.\label{eq:transfer_greater}
\end{eqnarray}
Equation~\ref{eq:transfer_greater} shows that the number of unmatched elements of $I$ is greater than the number of unmatched elements of $J$ and hence, for each $j\in \left(\cup_{y\in\Delta^-} U_y\right)\cup (\psi'(C_x)\cap J)$, we set $\mu(j)$ to some unique $i\in \left(\cup_{y\in\Delta^+} U_y\right)\cup (\psi'(C_x)\cap I)$. However, $\delta_x$ elements from $C_x\cap I$ remain unmatched and we assign them to $U_x$.
 \[U_x = \left(\left(\underset{y\in\Delta^+}{\cup} U_y\right) \cup \psi'(C_x)\right)\setminus\mu^{-1}(C_x\cap J).\] 
\item  Similarly, if $\delta_x<0$, Equation~\ref{eq:delta} implies
\begin{eqnarray}
\sum_{y\in\Delta^+}|U_y| + |\psi'(C_x)\cap I| < \sum_{y\in\Delta^-}|U_y| + |\psi'(C_x)\cap J|.\label{eq:transfer_lesser}
\end{eqnarray}
Equation~\ref{eq:transfer_lesser} shows that the number of unmatched elements of $I$ is less than the number of unmatched elements of $J$, and thus, for each $i\in \Delta^+\cup (\psi'(C_x)\cap I)$, we assign $\mu^{-1}(i)$ to some unique $j\in \Delta^-\cup (\psi'(C_x)\cap J)$. However, $\delta_x$ elements from $C_x\cap J$ remains unmatched and we assign them to $U_x$,
\[U_x =  \left(\left(\underset{y\in\Delta^-}{\cup} U_y\right) \cup \psi'(C_x)\right)\setminus\mu^{-1}(C_x\cap I).\] 
\end{itemize}
\end{enumerate}

The one-to-one property of $\mu$ is explicitly maintained during its construction. We now prove that all the elements of $J$ are matched. Let $H:=\{C_h,\ldots,C_\ell\}$ be the set of highest level categories, i.e., for each $C\in H$, $\Psi^{-1}(C)=\emptyset$. We create subgroups of $H$ depending on their $\delta$ values: 
$\Delta^+=\{y: C_y\in H\ \mbox{ and }\ \delta_y>0\}$, 
$\Delta^-=\{y: C_y\in H\ \mbox{ and }\ \delta_y<0\}$, and 
$\Delta^0=\{y: C_y\in H\ \mbox{ and }\ \delta_y=0\}$. Since $I = \underset{C\in H}{\cup} \ (C\cap I)$, $J = \underset{C\in H}{\cup} \ (C\cap J)$, and $|I|\geq |J|$, the following holds true: 
\begin{eqnarray}
&&\displaystyle \sum_{y\in \{h,\ldots,\ell\}}(|C_y\cap I|-|C_y\cap J|)\geq 0\nonumber\\
&\Rightarrow &\displaystyle \sum_{y\in \{h,\ldots,\ell\}}\ \delta_y\geq 0\nonumber\\
&\Rightarrow &\displaystyle \sum_{y\in \Delta^+}\delta_y + \sum_{y\in \Delta^-}\delta_y + \sum_{y\in \Delta^0}\delta_y \geq 0\nonumber\\
&\Rightarrow &\displaystyle \sum_{y\in \Delta^+}\delta_y \geq \sum_{y\in \Delta^-}(-\delta_y) \nonumber\\
&\Rightarrow &\displaystyle \sum_{y\in \Delta^+}|U_y| \geq \sum_{y\in \Delta^-}|U_y| \label{eq:transfer_top}
\end{eqnarray}
Equation~\ref{eq:transfer_top} implies that the number of unmatched elements of $I$ is greater than the number of unmatched elements of $J$, and thus, for each $j\in \left(\cup_{y\in\Delta^-} U_y\right)$, we assign $\mu(j)$ to some unique $i\in \left(\cup_{y\in\Delta^+} U_y\right)$. Since all elements belonging to categories in the set $\Delta^-$ get matched, we have that all the elements of $J$ is matched by the end of the procedure.

Finally, we show that swapping any $j \in J$ with $\mu(j) \in I$ leads to independent subsets, i.e., $I':= (I\cup \{j\})\setminus \{\mu(j)\}\in \mathcal{I}$ and $J':=(J\cup \{\mu(j)\})\setminus \{j\}\in \mathcal{I}$. Let $j\in J$ be an element which is mapped to $i\in I$ while processing the category $C_x$. Observe that both $i$ and $j$ belongs to category $C_x$. First, we consider the case when $i, j\in\psi'(x)$. Swapping $i$ and $j$ retains the total number of elements in category $C_x\cap I$ and $C_x\cap J$ and does not violate the threshold constraints. Next, we consider the case when $j\notin\psi'(x)$ (i.e., $j\in C_y\subset C_x$). Let $\{C_{1'},C_{2'},\ldots, C_{h'}, C_x\}$ denote the unique category-hierarchy for $j$, i.e., $j\in C_{1'}\subset C_{2'}\subset \ldots C_{h'} \subset C_x$. Since element $j$ remained unmatched till $C_x$, it must have propagated from $C_{1'}$ to $C_x$. This is possible only if, for each $y\in\{1',2'\ldots, h'\}$, $|C_y\cap I|-|C_y\cap J|=\delta_y<0$. Since, $|C_y\cap I|<|C_y\cap J|\leq k_y$, adding $j$ to $I$ do not violate the threshold $k_y$ for any $y\in\{1',2'\ldots, h'\}$. Now considering $i\notin\psi'(x)$, let $\{C_{1'},C_{2'},\ldots, C_{h'}, C_x\}$ denote the category-hierarchy for $i$. If $i$ remained unmatched till $C_x$, then it propagated from $C_1'$ to $C_x$ and hence, for each $y\in\{1',2'\ldots, h'\}$, $i\in U_y$ and $|C_y\cap I|-|C_y\cap J|=\delta_y>0$. Since $|C_y\cap J|<|C_y\cap I|\leq k_y$, adding $i$ to $J$ does not violate the threshold $k_y$ for any $y\in\{1',2'\ldots, h'\}$. Finally, the threshold constraint holds for $C_x$, since $|C_x\cap ((I\setminus\{i\})\cup \{j\})|=|C_x\cap I|\leq k_x$ and $|C_x\cap ((J\setminus\{j\})\cup \{i\})|=|C_x\cap J|\leq k_x$. Hence, the independence property is maintained. 


\end{proof}

Next we provide an efficient algorithm for finding an $\EFone$ allocation under identical, additive valuations and laminar matroid constraints. In particular, Theorem~\ref{theorem:identical_algo} provides a constructive proof of existence for $\EFone$ allocations in the current setup.

{
	\begin{algorithm}[!ht]
	{
		{{\bf Input :} A fair division instance $\left\langle[m],[n],v(\cdot),\mathcal{M}\right\rangle$ with identical, additive valuation, $v(\cdot)$, along with an oracle access to the underlying laminar matroid $([m],  \mathcal{I})$. \\
		{\bf Output:} A feasible $\EFone$ allocation.}
		\caption{SWAP ($\ALG$~\ref{alg:identical})}
		\label{alg:identical}
		\begin{algorithmic}[1]
			\STATE Initialize partition $\mathcal{A}=(A_1, \ldots, A_n)$ such that $\cup_{i\in[n]} A_i = [m]$ and $A_i\in\mathcal{I}$ for all $i\in[n]$.
			\STATE	Set $\ell = \displaystyle \argmin_{i \in [n]} v(A_i)$ \\ \COMMENT{$\ell$ is the index of the least valued bundle in $\mathcal{A}$}.
			\STATE Set $\mathcal{V} = \{ i \in [n] \mid v(A_i\setminus\{g\})>v(A_\ell) \mbox{ for all } g\in A_i\}$.  \\ \COMMENT{$\mathcal{V}$ is the set of indices $i$ such that the bundle $A_i$ violates $\EFone$ with respect to $A_\ell$}
			\WHILE{$\mathcal{V} \neq \emptyset$}
				\STATE Set $h = \displaystyle \argmax_{i \in \mathcal{V}} v(A_i)$ \\ \COMMENT{$h$ is the index of the highest valued bundle in $\mathcal{V}$}
				\IF{$|A_\ell|<|A_h|$}	
					\STATE Set $g = \displaystyle \argmax_{e \in A_h : A_\ell \cup \{ e \} \in \mathcal{I}} v(e)$ \\ \COMMENT{By the independent set exchange property of matroids, such a good $g$ is guaranteed to exist} 
					\STATE Update $A_h \leftarrow A_h \setminus \{g\}$ and $A_\ell  \leftarrow A_\ell \cup \{g\}$.
				\ELSE
					\STATE Let $\mu$ be the map obtained by applying Lemma~\ref{lemma:matching} to independent sets $A_\ell$ and $A_h$ \\ \COMMENT{Here, $|A_\ell| \geq |A_h|$ and we have $\mu: A_h \mapsto A_\ell$}.
					 
					\STATE Find a good $ g \in A_h$ such that $v(g) - v(\mu(g)) \geq \frac{1}{m^2} v(A_h)$ \\ \COMMENT{The existence of such a good is established in Theorem~\ref{theorem:identical_algo}}
						\STATE Update $A_h \leftarrow (A_h \setminus \{g\})\cup \{ \mu(g) \}$ and $A_\ell \leftarrow (A_\ell \setminus \{ \mu(g) \})\cup \{g\}$.
				\ENDIF
			\STATE	Update $\ell \leftarrow \displaystyle \argmin_{i \in [n]} v(A_i)$ and $\mathcal{V} \leftarrow \{ i \in [n] \mid v(A_i\setminus\{g\})>v(A_\ell) \mbox{ for all } g\in A_i\}$. 
			\ENDWHILE 			
			\STATE Return $\mathcal{A}=(A_1,\ldots,A_n)$.
		\end{algorithmic}}
	\end{algorithm}
}

\begin{theorem}\label{theorem:identical_algo}
Given any fair division instance $\left\langle[m],[n],(v_i)_{i\in[n]}, \mathcal{M}\right\rangle$ with additive, identical valuations and laminar matroid constraint, $\mathcal{M} \neq \emptyset$, Algorithm \ref{alg:identical} finds an $\EFone$ allocation in polynomial time.
\end{theorem}

\begin{proof}
The work of Gabow and Westermann~\cite{gabow1992forests} provides an efficient algorithm to compute a partition $\mathcal{A}=(A_1,\ldots,A_n)$ of the  $m$ goods such that $A_i \in \mathcal{I}$, for all $i \in [n]$. Algorithm~\ref{alg:identical} starts with such a partition and then iteratively reallocates goods between the bundles till an $\EFone$ allocation is obtained. The algorithm maintains the invariant that at any point of time the partition in hand is in $\mathcal{F}$. That is, throughout the execution of the algorithm the bundles allocated to the agents remain independent.

Let $\mathcal{A}=(A_1, A_2, \ldots, A_n)$ be a partition considered by the algorithm during its execution and let $\ell$ denote the index of the least valued bundle in $\mathcal{A}$, i.e., $\ell := \argmin_{i \in [n]} v(A_i)$. Also, let $\mathcal{V}(\mathcal{A})$ denote the set of indices of the bundles that violate $\EFone$ with respect to the least valued bundle $A_\ell$:
\begin{align*}
\mathcal{V}(\mathcal{A}) & := \{ i \in [n] \mid v(A_i) -v(g)> v(A_\ell) \mbox{ for all } g \in A_i\}.
\end{align*}

Note that allocation $\mathcal{A}$ is $\EFone$ iff $\mathcal{V}(\mathcal{A})=\emptyset$. We will use $A_h$ to denote the highest valued bundle which violates the $\EFone$ condition with respect to $A_\ell$, i.e., $h := \argmax_{i \in \mathcal{V}(\mathcal{A}) } \ v(A_i)$. By definition of the the violating set, we have,
\begin{align}\label{eq:EFone_H_L}
v( A_h \setminus \{g\}) & > v(A_\ell) \quad \mbox{for all } g\in A_h
\end{align}

Algorithm~\ref{alg:identical} aims to reallocate good(s) between $A_h$ and $A_\ell$ such that the resulting bundles remain independent. We will prove (via case analysis) that the intended reallocation is always possible and it results in either (i) a decrease in the cardinality of $A_h$ by one, or (ii) a drop in the value of $A_h$ by a multiplicative factor of  $\left(1-\frac{1}{m^2}\right)$. 

\begin{enumerate}
\item[] \textbf{Case (i)}: $|A_\ell| < |A_h|$. Since $A_\ell, A_h \in\mathcal{I}$, by the independence set exchange property, there exists a good $g \in A_h$ such that $A_\ell \cup\{g\} \in\mathcal{I}$. Algorithm \ref{alg:identical} transfers such a good $g$ from $A_h$ to $A_\ell$, and in this case the cardinality of $A_h$ drops by one. 
\item[] \textbf{Case (ii)}: $|A_\ell| \geq |A_h|$. Let $t:=|A_h|$ and let $A_h=\{ g_z\}_{z \in [t]}$. Without loss of generality, assume that the goods in $A_h$ are indexed such that $v(g_1)\geq v(g_2)\geq \ldots\geq v(g_t)$. Lemma~\ref{lemma:matching} provides $t$ distinct pair of goods $\{ (g_z, \mu(g_z)) \in (A_h, A_\ell) \mid z \in [t]\}$ such that  $(A_h \setminus\{g_z\})\cup\{ \mu(g_z)\}$ and $(A_\ell \setminus\{ \mu(g_z) \})\cup\{g_z\}$ are independent for all $z \in [t]$.  We show that among these pairs there exists one, $(g_x, \mu(g_x))$, which satisfies $$v(g_x)-v( \mu(g_x)) \geq \frac{1}{t^2} v(A_h).$$

The fact that $A_h$ and $A_\ell$ violate the $\EFone$ condition (\ref{eq:EFone_H_L}) for the good $g_1$, in particular, implies  $v(A_h) - v(A_\ell)> v(g_1) $. Since $v(A_h) = \sum_{z=1}^t v(g_z)$ and $v(A_\ell) > \sum_{z=1}^t v(\mu(g_z))$, we have $\sum_{z=1}^t\left(v(g_z)-v( \mu(g_z)) \right)> v(g_1) $. An averaging argument ensures that there exists an index $x$ such that 
\begin{align}
v(g_x)-v( \mu(g_x)) & > \frac{1}{t} v(g_1) \nonumber \\
& \geq \frac{1}{t^2} v(A_h) \qquad \text{($g_1$ is the highest valued good in $A_h$)} \label{eq:lower_bound_on_delta} 
\end{align}
Since $t \leq m$, the selected pair $(g_x, \mu(g_x))$ satisfies the required property $v(g_x)-v( \mu(g_x)) \geq \frac{1}{m^2} v(A_h)$. Therefore, 
\begin{align*}
v(A_h) - v(g_x) + v(\mu(g_x)) & < \left(1-\frac{1}{m^2}\right) v(A_h)
\end{align*}

Lemma~\ref{lemma:matching} ensures that the bundles $(A_h \setminus \{g_x\}) \cup \{ \mu(g_x) \}$ and $(A_\ell \setminus \{ \mu(g_x) \} ) \cup \{ g_x\}$ are independent. Hence, whenever $|A_h| \leq |A_\ell|$, Algorithm~\ref{alg:identical} can swap goods between $A_h$ and $A_\ell$ such that the bundles remain independent and the value of the bundle $A_h$ drops by a multiplicative factor of at least $\left(1 - \frac{1}{m^2} \right)$.
\end{enumerate}

Consider an iteration wherein the algorithm updates partition $\mathcal{A}=(A_1, A_2, \ldots, A_n) \in \mathcal{F}$ to obtain partition $\mathcal{A}'=(A'_1, A_2, \ldots, A'_n) \in \mathcal{F}$. This update decreases the valuation of the bundle $A_h$ and increases the valuation of $A_\ell$: $v(A'_h) < v(A_h)$ and $v(A'_\ell) > v(A_\ell)$. For all the remaining bundles the valuation remains unchanged:  $v(A'_i) = v(A_i)$ for all $i \in [n] \setminus \{h, \ell\}$. 

Since $A_h$ violates $\EFone$ and the algorithms reallocates one good from $A_h$, we have $v(A'_h) \geq v(A_\ell)$. Therefore, $\min_{i \in [n] } v(A'_i) \geq \min_{i \in [n]} v(A_i)$. Using this observation and the fact that $A_\ell$ is the least valued bundle in $\mathcal{A}$, we get that the algorithm's update does not create new violators with respect to the least valued bundle: $\mathcal{V}(\mathcal{A}') \subseteq \mathcal{V}(\mathcal{A})$. 

This containment ensures that if an index $i$ is dropped from the violating set $\mathcal{V}(\mathcal{A})$ then it is never included back in the set in any later iteration. Next, we will show that an index can be the highest valued bundle, i.e., $h$, in $\mathcal{V}(\mathcal{A})$ for at most a polynomial number of iterations. This, overall, implies that the algorithm finds an $\EFone$ allocation in polynomial time. 

Note that every time an index $i$ is selected as the highest valued bundle in $\mathcal{V}(\mathcal{A})$ either the cardinality of $A_i$ decreases by one or its valuation drops  by a multiplicative factor of  $\left(1-\frac{1}{m^2}\right)$. Furthermore, while $A_i$ continues to violate the $\EFone$ condition with respect to the least valued bundle its value cannot increase in an update (only the least valued bundle experiences an increase in valuation). Hence, after $i$ is selected as the highest valued bundle in the violating set, say, $m^2 k + m$ times, $v(A_i)$ is at most $\left(1 - \frac{1}{m^2} \right)^{m^2 \ k}  \leq \frac{1}{e^k}$ times its initial value. Therefore, for a polynomially large $k$, we get that any index $i$ cannot be the highest valued bundle in the violating set more than $m^2 k + m$ times. Overall, this establishes that the algorithm terminates after a polynomial number of iterations and the stated claim follows.
\end{proof}

\section{Conclusion and Future Work}
This paper extends the active line of work on fair division of indivisible goods and shows that fairness guarantees are not lost by imposing cardinality constraints. In particular, we show that $\EFone$ allocations are guaranteed to exist even under cardinality constraints. Note that, though the round-robin method~\cite{caragiannis2016unreasonable} and cycle-elimination algorithm \cite{lipton-envy-graph} efficiently find $\EFone$ allocations in the unconstrained setting, these algorithms can lead to an allocation which does not satisfy the cardinality constraints. In this paper we bypass this issue by combining the round-robin method with envy graphs in an interesting manner. 

The universality of $\EFone$ is further strengthened by our result which shows that if the agents' valuations are identical (and additive), then fair ($\EFone$) allocations exist even under laminar matroid constraints. Establishing such a guarantee---along with computational results---for heterogeneous valuations (additive and beyond) remains an interesting direction of future work.\footnote{In and of itself, the approach used in Section~\ref{sec:identical} does not establish such a guarantee: one can construct instances where the agents have additive (but different) valuations and the allocation that maximizes the Nash social welfare, subject to cardinality constraints, is not $\EFone$.} 

We also show that, even under cardinality constraints, approximate maximin fair allocations always exist and can be computed efficiently. This result is obtained by reducing the constrained version of the problem (with additive valuations) to an unconstrained one (with submodular valuations). This reduction might be useful for addressing other constrained, fair division problems.   

Going forward it would be quite relevant to study fair division, both unconstrained and constrained, among strategic agents. 

\subsubsection*{Acknowledgments.} The authors thank Fahad Panolan for extremely useful discussions which lead to the proof of Theorem~\ref{theorem:identical_algo}. Arpita Biswas gratefully acknowledges the support of a Google PhD Fellowship Award. Siddharth Barman was supported by a Ramanujan Fellowship (SERB - {SB/S2/RJN-128/2015}). 

\bibliographystyle{alpha}
\bibliography{MMSbib}

\end{document}